\documentclass[a4paper]{elsarticle}

\makeatletter
\def\ps@pprintTitle{%
 \let\@oddhead\@empty
 \let\@evenhead\@empty
 \def\@oddfoot{}%
 \let\@evenfoot\@oddfoot}
\makeatother

\usepackage{amsmath}
\usepackage{amsfonts}
\usepackage{amssymb}
\usepackage{amsthm}
\usepackage{graphicx}
\usepackage{color}
\usepackage[colorlinks=true, urlcolor=cyan, citecolor=magenta, breaklinks=true]{hyperref}
\usepackage{enumitem}
\usepackage{bm}
\usepackage{mathtools}
\usepackage{leftidx}
\usepackage{isomath}
\usepackage[usenames,dvipsnames]{pstricks}
\usepackage{epsfig}
\usepackage{pst-grad} 
\usepackage{pst-plot} 
\usepackage{textcomp} 
\usepackage{nicefrac}
\usepackage{wrapfig}
\usepackage{algpseudocode}
\usepackage{framed}
\usepackage[hyphenbreaks]{breakurl}
\DeclarePairedDelimiter{\ceil}{\lceil}{\rceil}

\newcommand{\remove}[1] {}
\newtheorem{definition}{Definition}
\newtheorem{proposition}{Proposition}
\newtheorem{lemma}{Lemma}
\newtheorem{corollary}{Corollary}
\newtheorem{theorem}{Theorem}

\newtheorem{remark}{Remark}

\newtheorem*{summary*}{Summary of results}
\newtheorem*{importantrmk*}{Important remark}

\usepackage{xspace}

\newcommand{\reco}{{\sc Re\-color}\xspace}
\newcommand{\cola}{{\sc Color}\xspace}
\newcommand{\valc}{{\sc Color\-Val}\xspace}

\newcommand{\C}{\mathcal{C}} 
\newcommand{\seq}{\mathcal{S}}

\newcommand{\hk}{2\hat{k}}

\newlength{\pgmtab}
\begin{document}

\title{Acyclic Edge Coloring\\ through the Lov\'{a}sz Local Lemma\tnoteref{ack}
}
\tnotetext[ack]{Research co-financed by the European Union (European Social Fund ESF) and Greek national funds through the Operational Program ``Education and Lifelong Learning" of the National Strategic Reference Framework (NSRF) Research Funding Program: ARISTEIA II. 
}

\author[uoa,cti]{{Ioannis Giotis}}
\ead{igiotis@cs.upc.edu}
\author[uoa,cti]{Lefteris Kirousis\corref{correspondingauthor}}
\cortext[correspondingauthor]{Corresponding author}
\ead{lkirousis@math.uoa.gr} 
\author[uoa]{{Kostas I. Psaromiligkos}}
\ead{kostaspsa@gmail.com}
\author[uoa,cnrs,cti]{{Dimitrios~M.~Thilikos}}
\ead{sedthilk@thilikos.info}

\address[uoa]{Department of Mathematics,  National \& Kapodistrian University of Athens, Greece}
\address[cnrs]{AlGCo project-team, CNRS, LIRMM, Montpellier, France}
\address[cti]{Computer Technology Institute \& Press ``Diophantus", Patras, Greece}

\begin{abstract}
We give a  probabilistic analysis of a Moser-type algorithm for the   Lov\'{a}sz Local Lemma (LLL),   adjusted to search  for acyclic edge colorings of a graph. We thus improve the best known upper bound to acyclic chromatic index, also obtained by analyzing a similar algorithm, but  through the entropic method (basically counting argument). Specifically we show that a graph with maximum degree $\Delta$ has an acyclic proper edge coloring with at most $\lceil 3.74(\Delta-1)\rceil+1 $ colors, whereas, previously,  the best  bound was $4(\Delta-1)$. The main contribution of this work is that it comprises a  probabilistic analysis of a Moser-type algorithm  applied to events pertaining to dependent variables.  
\end{abstract}


\maketitle

\section{Introduction and the basic  algorithm}\label{sec:colors}

Let $G = (V,E)$ be a (simple) graph with $l$ vertices and $m$ edges. The  chromatic index of $G$ is the least number of colors needed to {\em properly} color its  edges, i.e., to color them so that no adjacent edges get the same color.  If $\Delta$ is the maximum degree of $G$, it is known that its chromatic index is either $\Delta$ or $\Delta+1$ (Vizing \cite{vizing1965critical}).  

 A cycle of $G$ of length $s$ is a sequence $v_i, i =0, \ldots, s-1$  of distinct vertices so that $\forall i = 0, \ldots s-1$, $v_{i}$ and $v_{i+1 \pmod s}$ are connected by an edge.  The {\em acyclic} chromatic index of $G$ (ACI)  is defined as the least number of colors needed to properly color the edges of $G$ so that no   cycle is {\em bichromatic}, i.e., so that there is no cycle  whose  edges are properly colored with only two colors.  Notice that in any properly colored graph, any cycle of odd  length is necessarily at least  {\em trichromatic}, i.e., its edges have three or more colors. It has been conjectured (J. Fiam\v{c}ik \cite{fiam} and Alon et al. \cite{alon2001acyclic}) that the acyclic chromatic index of any graph with maximum degree $\Delta$ is at most  $\Delta +2$.  A  number of successively tighter upper bounds to the acyclic chromatic index have been provided in the literature. Most recently,  Esperet and Parreau \cite{DBLP:journals/ejc/EsperetP13} proved that the acyclic chromatic index is at most  $4(\Delta -1)$.  Their proof makes use of  the technique  of
Moser and Tardos \cite{moser2010constructive}  that constructively proves the L\'{o}vasz Local Lemma (a technique which became known as the ``entropy compression method" \cite{TT}). 
An approach using the entropy compression method was also used for the vertex analogue of 
 the edge chromatic number by Gon\c{c}alves et al. \cite{1406.4380}.
 
In this work, 
we modify the technique used by Esperet and Parreau \cite{DBLP:journals/ejc/EsperetP13}  in that for a Moser-type  edge coloring algorithm, we use as tool of analysis  
the approach we described in Section 2 of  \cite{giotisanalco}. Namely, instead of an essentially counting argument as used in the entropy compression method,  we give a  probabilistic analysis that yields an upper bound of $\lceil 3.74(\Delta-1) \rceil+1 $ for the acyclic chromatic index, improving over $4(\Delta-1)$ in \cite{DBLP:journals/ejc/EsperetP13} (in contrast to the paper by Moser and Tardos \cite{moser2010constructive}, a  probabilistic analysis was used in the original paper of Moser \cite{moser2009constructive}; see the elegant exposition   by Spencer in \cite{Sp}). The present paper can be read independently of \cite{giotisanalco}.
 
An interesting aspect of this application is that the edge colors  to which the  ``undesirable" events of LLL refer to are not probabilistically independent. This dependence introduces certain conceptual  difficulties, some of which are, we believe, non-trivial (see the proof of Lemma \ref{lem:ourapproach} and the preceding remarks).  The randomized algorithm that we deal with  allows   the  freedom, when coloring an edge,  to  make a selection, uniformly at random,    from a guaranteed minimum  number of available colors. However,  the ``guarantee" of a minimum number of available colors is valid only  if any conditioning refers only  to colors previously assigned.
To handle the probabilistic analysis of such  a randomized algorithm without introducing posterior probabilities, which would render the analysis unmanageable,  we put all  events referring to colors that edges have in chronological order according to the instant these edges got their current color.   We consider this approach of handling dependent events in  constructive proofs of LLL, rather than just  the improvement of the coefficient of the upper bound from 4 to 3.74,  as the  essential aspect of the contribution of this work (nevertheless, see the discussion  in Section \ref{sec: discussion} for possible further numerical improvement).

We also get improved numerical results with respect to graphs with bounded girth, some specific values of which are sampled  in Figure~\ref{fig:results}. 
\begin{figure}[h]
\centering
\begin{tabular}{|c|c|c|}\hline
Girth & Number of colors & Previously known~\cite{DBLP:journals/ejc/EsperetP13} \\
\hline - & $3.731(\Delta-1)+1$ & $4(\Delta-1)$\\
\hline 7 & $3.326(\Delta-1)+1$ & $3.737(\Delta-1)$\\
\hline 53 & $2.494(\Delta-1)+1$ & $3.135(\Delta-1)$\\
\hline  219 & $2.323(\Delta-1)+1$ & $3.043(\Delta-1)$\\ \hline
\end{tabular}
\caption{Our results}
\label{fig:results}
\end{figure}
\remove{For graphs with girth at least 7, our technique yields the bound $\lceil X.XX(\Delta-1) \rceil +1$ compared to the bound $\lceil 3.74(\Delta-1)\rceil$ in \cite{DBLP:journals/ejc/EsperetP13}, and for graphs with girth at least 53, we get the bound $\lceil X.XX(\Delta-1) \rceil +1$, whereas the corresponding one  in \cite{DBLP:journals/ejc/EsperetP13} is $\lceil 3.14(\Delta-1)\rceil $. }
%

Below, to facilitate notation, we call a proper edge-coloring {\em $s$-acyclic} if it contains no bichromatic cycle  of length $s$ or less. We  call the corresponding graph parameter the {\em s-acyclic chromatic index}.

We start by mentioning the  following fact, proved in Esperet and Parreau~\cite{DBLP:journals/ejc/EsperetP13}:
\begin{lemma}[Esperet and Parreau~\cite{DBLP:journals/ejc/EsperetP13}]\label{sufficientcolors} At any step of any successive coloring of the edges of a graph, there are at most 
$2 (\Delta-1)$ colors that should be avoided in order to produce a 4-acyclic coloring. 
\end{lemma}
\begin{proof}[Proof Sketch]
Notice that for each edge $e$, one has to avoid the colors of all edges  adjacent to $e$, and moreover  for each pair of homochromatic (of the same color) edges $e_1, e_2$  adjacent to $e$ at different endpoints (which contribute one to the count of colors to be avoided), one has also to avoid the color of the at most one edge $e_3$ that together with $e,e_1, e_2$ define a cycle of length 4. So the total count of colors to be avoided adds up to $2 (\Delta-1)$.\end{proof}
Assume now that we have $K=\ceil{(2 + \gamma)(\Delta-1)} +1 $ colors, where $\gamma$ is a nonnegative constant to be computed. 

We assume below that  the edges of the graph,  and its cycles, are ordered according to  fixed a priori orderings.
 
Notice that for each cycle of the graph there are two  consecutive traversals of its edges. It does not matter which we use, but for concreteness when we start the traversal from an edge $e$, the next edge  to be traversed is the least one from the two adjacent to $e$.  We call this traversal ``positive". 

Also in an even length cycle, the edges can be partitioned into two subsets of equal cardinality, the  elements of each of which have pairwise odd distance. These sets are called equal parity sets. If we color such a cycle with two colors so that no adjacent edges get the same color, then the equal parity sets are the {\em monochromatic} sets, i.e. the sets whose respective elements get the same color.

Consider now the algorithm  given in Figure \ref{fig:col}.  

\begin{figure}[h]
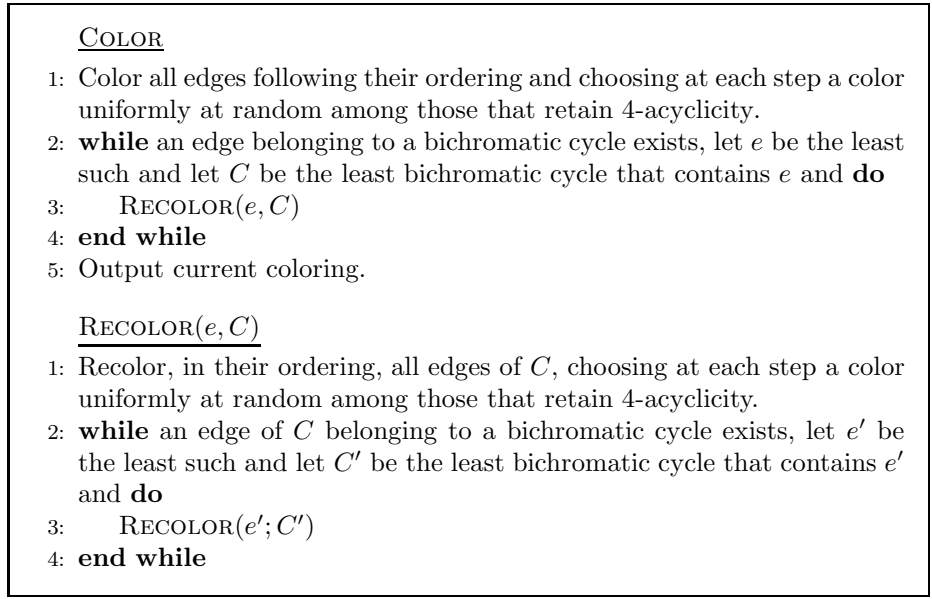

\begin{framed}
\begin{algorithmic}[1]
\Statex{\underline{{\cola}}}\smallskip
\State Color all edges following their  ordering and choosing at each step a color uniformly at random among those that retain   \mbox{4-acyclicity}.\label{line:color_initial}
\While{an edge  belonging to a  bichromatic cycle exists,  let $e$ be the least such  and let  $C$ be the least bichromatic cycle that contains $e$ and} \label{line:color_check}  
 \State \Call{Recolor}{$e, C$} \label{line:recolor_root}
\EndWhile
\State Output current coloring.
\end{algorithmic}
\medskip
\begin{algorithmic}[1]
\Statex{\underline{{\reco}($e,C$)}}\smallskip
\State 
 Recolor, in their ordering, all edges of $C$,   choosing at each step a color uniformly at random  among  those that retain 4-acyclicity.
\While {an edge of $C$  belonging to a  bichromatic cycle exists,  let $e'$ be the least such  and let  $C'$ be the least bichromatic  cycle that contains $e'$  and}    \State \Call{Recolor}{$e';C'$} \label{line:recolor_internal}\EndWhile
\end{algorithmic}
\end{framed}
\caption{ {\textit{The coloring  algorithm}}}\label{fig:col}
\end{figure}
A {\em root}  call of \reco is a call invoked in line~\ref{line:recolor_root} of \cola. A {\em phase} is the collection of steps involved in the execution  of a call of \reco (root or not) or the collection of steps executed in line~\ref{line:color_initial} of \cola. The latter phase is called the initial phase  (phases are nested).  The time complexity of  \cola is expressed in terms of the number of invocations (root or not)  of \reco, i.e. the number of phases.

Notice that in this algorithm,  each non-root invocation of  \reco chooses  a cycle sharing an edge with the cycle of the calling \reco. So this algorithm is in the spirit of the algorithm in the original paper by Moser \cite{moser2009constructive};  in  the subsequent algorithm by Moser and Tardos \cite{moser2010constructive},   the event chosen at the end of each phase does not necessarily share a variable with  the event of the previous phase.

The following follows easily from line~\ref{line:color_check} of \cola.

\begin{lemma}\label{stops}
\cola  outputs an acyclic edge coloring {\em if it ever stops}. 	
\end{lemma}
The main result of this paper, Theorem \ref{colthm}, states that if there are at least $ 3.74(\Delta-1)+1 $ available colors, then the probability that \cola lasts for at least $n$ phases is $<1$ for sufficiently large $n$ (actually, Theorem \ref{colthm} states that this probability is subexponential), therefore there is an acyclic edge coloring.

\begin{lemma}\label{lem:colorprogress}
Consider an arbitrary call of {\em \reco}($e,C$). Let $\mathcal{E}$ be the set of edges   that do not belong to a   bichromatic cycle at the beginning of this call together with the set of edges of $C$. Then, if the call terminates, the edges  in $\mathcal{E}$ do not belong to any bichromatic cycle at the end of {\em \reco}($e,C$). \end{lemma}
\begin{proof}  
Consider an edge $e' \in \mathcal{E}$. Assume first that $e'$ does not belong to $C$. 
Then, if $e'$ is in a bichromatic
    cycle $C'$ after the execution of {\reco}($e,C$), that cycle must have
    become bichromatic during the execution.
Therefore there is a cycle $C^-$  and an edge of it $e^-$ such that the process {\reco}($e^-; C^-$) was called as a recursive call of { \reco}($e, C$) and   $C'$ and $C^-$ share an edge and $C'$ became bichromatic during  {\reco}($e^-; C^-$). 
But the call { \reco}($e^-;C^-$) will not 
terminate until all    edges in $C^-$ do not belong to a bichromatic cycle.  Assuming { \reco}($e,C$) terminates, {\reco}($e^-;C^-$) must have also terminated. So   $C'$ cannot be bichromatic at the end of {\reco}($e^-;C^-$).  
The same argument can be reapplied every time $e$ changes color and becomes an element of  a bichromatic  cycle during { \reco}($e,C$). 

If, on the other hand,  $e' \in C$, then by definition $e'$ does not belong to any bichromatic cycle at the end of {\reco}($e,C$). \end{proof}
As an immediate corollary we get:
\begin{corollary}\label{cor:colorprogress} The cycles of   root calls of \reco  have pairwise  distinct sets of edges. Therefore the number of root phases is  at most $m$,   the number of edges of the graph. 
	
\end{corollary}

We define labeled forests to be rooted forests  whose nodes are labeled with pairs $(e, C)$, where $e$ is an edge, called the edge-label, and $C$ is an even  cycle of half-length $\geq 3$ that contains $e$, called the cycle-label. We consider such labeled forests as ordered by ordering the set of roots and each set of siblings (children of the same parent) according to the order of  their respective edge-label.  We also define:
 \begin{definition}\label{def:colorfeas}
A labeled   forest is called {\em feasible} if 
\begin{enumerate}[label = \roman*.]
\item \label{itmi} Pairwise,  the cycle-labels of the roots of the trees do not share an edge,
\item \label{itmii} pairwise the  cycle-labels of the children of every node do not share an edge,
\item \label{itmiii} if $C$ is the cycle-label of a node $u$, the edge-label of any child of $u$ belongs to $C$. 
\end{enumerate}
\end{definition}

Feasible forests are intended to represent the structure of recursive calls to \reco. For technical reasons,   we  add to a feasible forest some new leaves to which we assign  just an edge-label (their cycle-label can be taken to be an empty cycle):  First  we add  new  trees comprised of a root only, so that the set of edge-labels of all the roots of the trees of forest becomes  equal to the set of edges of the graph; second, we hang from each  original leaf $v$ of $\mathcal{F}$ as many new leaves as the edges of the cycle-label of $v$, and we label  them so that   the set of their  edge-labels  coincides with the set of edges of the cycle-label of $v$.  So  we assume in the sequel  that there are exactly $m$ roots whose edge-labels comprise the set of all edges,  and that any internal  node with a cycle-label $C$ of half-length $k$ has exactly $2k$ children whose edge-labels comprise the set of all edges in $C$.  Traversing the internal nodes of each tree of a feasible forest    in pre-order (depth-first),  and visiting the trees in the order of their roots' edge-labels, we obtain the forest's  label-sequence $$\mathcal{L} = (e_1, C_1), \ldots, (e_n, C_n)$$ (labels of leaves are not included in the label-sequence). 
 
\begin{definition}
The $n$-witness forest (or just the witness forest, when $n$ is clear from the context) of an execution  of \cola with at least $n$ phases is the feasible forest that results by creating one node per each of the $n$ invocations of \reco, labeling it by its argument, and structuring the trees of the forest  as the  calls of \reco appear in the recursion stack of  each root phase, i.e.  a node labeled with $(e_2,C_2)$ is a progeny of a node labeled with $(e_1, C_1)$ if {\em \reco}($e_2,C_2$) is called while \mbox{{\em \reco\!\!($e_1, C_1$)}} is executed (additional leaves as described in the previous paragraph are also added).   
\end{definition}

\section{A Bound for the Acyclic Chromatic Index}\label{sec:ACI}
Towards finding an upper    bound that  \cola lasts for at least   $n$ phases,
we consider  the   algorithm \valc in Figure \ref{fig:colvalalg} that takes as input an arbitrary  sequence $$\mathcal{S} = (e_1^1, e_1^2, k_1),\ldots, (e_n^1, e_n^2, k_n),$$ such that  for all $s=1, \ldots, n$,  $e_s^1$ and  $e_s^2$ are adjacent edges contained in some cycle of half-length $ k_s\geq 3$, and appear in this order in the cycle's positive traversal, when we start from $e_s^1$. We call such sequences {\em admissible sequences.}

\renewcommand{\labelenumi}{(\alph{enumi})}
\renewcommand{\theenumi}{(\alph{enumi})}
\algnewcommand{\myIf}[1]{\State\algorithmicif\ #1}
\algnewcommand{\myThen}[1]{\State\algorithmicthen\ #1}
\algnewcommand{\myEndIf}{\State \algorithmicend\ \algorithmicif}
\algnewcommand{\myElse}[1]{\State\algorithmicelse\ #1}
\newcommand\algotext[1]{\end{algorithmic}#1\begin{algorithmic}[1]}

\begin{figure}[h]
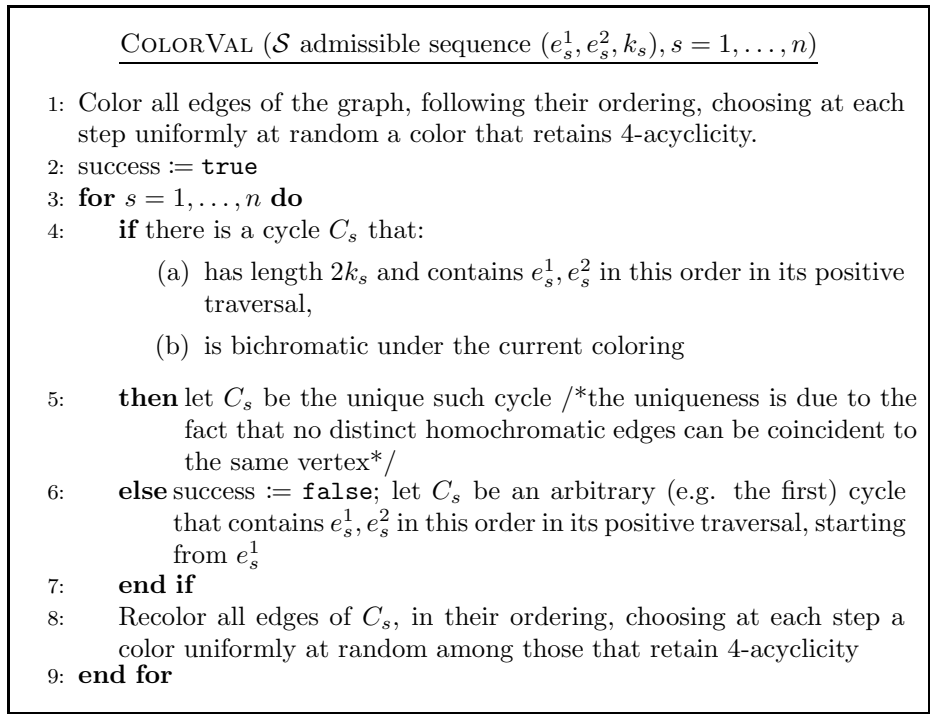

\begin{framed}
\begin{center}{\underline{\valc ($\mathcal{S} \text{ admissible sequence }   (e^1_s, e^2_s, k_s), s=1, \ldots, n)$}}\end{center}
\begin{algorithmic}[1]
\State Color all edges of the graph, following their ordering,  choosing at each step uniformly at random  a color that retains \mbox{4-acyclicity.}\label{line:firstphase}
\State $\text{success} \coloneqq \text{\tt true}$
\For{$s =1 ,\ldots, n$} \label{line:forloop}\label{line:nextphases}
   \myIf  {\parbox[t]{\dimexpr\linewidth-2.5em}{there is a  cycle $C_s$ that:
   \begin{enumerate} \item  has  length $2k_s$ and  contains $e^1_s, e^2_s$ in this order in its positive traversal, \item   is  bichromatic  under the current coloring  \end{enumerate} }} \label{line:cyclecond}  
   \myThen {\parbox[t]{\dimexpr\linewidth-3.5em} {let $C_s$ be the unique such cycle /*the uniqueness is due to the fact that no distinct homochromatic edges can be coincident to the same vertex*/ }}
   \myElse {\parbox[t]{\dimexpr\linewidth-3.5em}{$\text{success} \coloneqq \text{\tt false}$; let $C_s$ be an arbitrary (e.g. the first)  cycle that contains $e^1_s, e^2_s$ in this order in its positive traversal, starting from $e_s^1$ \strut}}   \myEndIf 
\State \parbox[t]{\dimexpr\linewidth-1.5em }{Recolor  all edges  of $C_{s}$, in their ordering,   choosing at each step a color uniformly at random among those  that retain 4-acyclicity }\label{line:s-phase}\EndFor  
\end{algorithmic}
\end{framed}
\caption{ \textit{The  coloring  validation algorithm}}\label{fig:colvalalg}
\end{figure}

Observe that \valc on input an arbitrary admissible sequence $\mathcal{S}$, with $|\mathcal{S}|=n$,   always produces as output a  uniquely defined sequence of cycles $\mathcal{C}= C_1, \ldots, C_n$ (but, alas,  may end up with success = {\tt false}).  The notion of ``phases" for \valc is defined in analogy to \cola, i.e. phases are comprised  of the number of steps   that are  involved with  the execution of line \ref{line:firstphase} (initial phase)  or with the execution    of a loop of line \ref{line:nextphases} of \valc (phase $s$).

Given a feasible forest $\mathcal{F}$ whose  label-sequence is $(e_s, C_s), s=1, \ldots,n$ {\em the corresponding admissible sequence} $\mathcal{S}_\mathcal{F}$ is obtained by letting  $e_s^1$ be  $e_s$, $e_s^2$ be the edge of $C_s$  following $e_s$ in $C_s$'s positive traversal and $k_s$ be the half-length of $C_s$, $s=1, \ldots, n$. Given an admissible sequence $\mathcal{S}$ let $\frak{F}_\mathcal{S}$ be the {\em class} of all feasible forests $\mathcal{F}$ such that $\mathcal{S}_\mathcal{F} = \mathcal{S}.$

\begin{proposition}\label{prop:proposition}  Given an admissible sequence ${\mathcal S}$
\begin{multline}\label{eq:proposition} 
\sum_{{\mathcal F} \in {\frak F}_{\mathcal S}} \Pr[ \text{ \cola executes with  witness forest } {\mathcal F}] \leq \\
	\Pr[\text{\valc is successful on input } {\mathcal S}].
\end{multline}

\end{proposition}
\begin{proof} We prove first that the probability  $P$   for at least one ${\mathcal F} \in {\frak F}_{\mathcal S}$ being a witness forest of \cola is bounded from above by the probability that \valc is successful on input  ${\mathcal S}$. 
For this it is sufficient to notice that if the random choices made by an execution of \cola that produces an arbitrary  ${\mathcal F} \in {\frak F}_{\mathcal S}$ 
are made by \valc on input ${\mathcal S}_{\mathcal F} = {\mathcal S}$, 
then \valc is successful. The result now follows by observing the the events that 
${\mathcal F}$ is a witness forest of \cola for various 
${\mathcal F} \in {\frak F}_{\mathcal S}$ are mutually exclusive, therefore $P$ can be written as the sum in the lhs of \eqref{eq:proposition}. \end{proof}

Given an admissible sequence ${\mathcal S} = (e_1^1, e_1^2, k_1), \ldots, (e_n^1, e_n^2, k_n) $, 
an arbitrary  sequence  $ {\mathcal C} = C_1, \ldots, C_n$  of cycles such that $C_s, s=1, \ldots, n,$ has length $2k_s$ and contains $e_s^1, e_s^2$,  in this order in its positive traversal, starting from $e_s^1$,   is called  a cycle-sequence {\em associated} with the admissible ${\mathcal S}$ (a cycle-sequence ${\mathcal C}$ associated with ${\mathcal S}$ need not necessarily be the cycle-sequence produced by \valc).

Given a cycle-sequence ${\mathcal C} = (C_s)_s$ associated with an admissible sequence ${\mathcal S} = (e_s^1, e_s^2, k_s)_s$,  let the other edges of $C_s$ in its positive traversal be $e_s^3, \ldots e_s^{2k_s}$ and let $v_s^i, i=1, \ldots, 2k_s,$ be the vertices of $C_s$  with $v_s^{i}, v_s^{i+1 \mod 2k_s}$ being the endpoints of $e_s^i$.  
Finally let $2 \hat{k} = \sum_{s=1}^n (2k_s -2)$.

For linguistic convenience, we call the edges $e_s^1, e_s^2$ the {\em first and second pivotal edges}  of $C_s$, respectively (these are common to all ${\mathcal C}$).   
 Let $e_s^{l^1_s}$ (respectively, $e_s^{l^2_s}$) be the edge of $C_s$ that gets the color it has at  phase $s$ of \valc   earliest among the edges of $C_s$  that have equal parity with $e_s^1$ ($e_s^2$, respectively). We call  $e_s^{l^1_s}, e_s^{l^2_s}$ the {\em early} edges of $C_s$. 
 
We call  {\em instants} of an  execution of \valc the discrete successive values $t=1, 2, \ldots$ of a time parameter at which \valc assigns a color to an edge. Reassigning the same color is assumed to take place at a new instant. 

Call the color an  edge has   immediately before the $s$'th repetition of loop \ref{line:forloop}  the edges's color at phase $s$. Also call an edge's  color at phase 1 the edge's {\em initial}  color. Given an edge $e$ and a phase number $s$, let $\mathrm{time}(s; e$) be the latest instant before phase $s$ starts that  $e$ was assigned the color it has at phase $s$, and let $\chi(s; e)$ be this color. 
\begin{remark}\label{rem:distincttime}
Notice that in general it is possible to have two distinct phases $s', s$ and  an edge $e$ so that  $\mathrm{time}(s'; e) = \mathrm{time}(s; e)$, because an edge may not be assigned a color for the duration of several phases; however if $C_{s'}$ is the cycle produced by an execution of \valc at phase $s'$,  $s' <s$, and $e$ belongs to $C_{s'}$     then $\mathrm{time}(s'; e) <  \mathrm{time}(s; e)$, because the cycles of the output cycle-sequence  are recolored at each repetition of loop \ref{line:forloop}. 	
\end{remark}

Let $c^1_s, c^2_s$ be the colors that the early edges   $e^{l^1_s}_s, e^{l^2_s}_s$, respectively, of a cycle $C_s$ have at phase $s$ i.e. $c^j_s = \chi(s; e_s^{l^j_s}), j=1,2$. 

\begin{remark}\label{rem:monty} To avoid a possible misinterpretation, let us stress here that we do {\em not} assume that $c^1_s, c^2_s$  are  a priori fixed  and that the   executions considered are conditioned  on 
 the early edges of $C_s$ taking the  colors $c^1_s, c^2_s$, respectively. Rather, given an {\em arbitrary}  execution,  $c^1_s, c^2_s$ denote   the colors the early edges    {\em will} have  at phase $s$.  Much like the   Monty Hall problem \cite{selvin1975problem},  where the assumption that  Monty Hall opens  box A    does not mean that the choice of the contestant is conditioned  on Monty Hall  always  opening box A. ``Box A"  is just a ``name" for whatever box Monty Hall will open. Its actual value  depends on the box selected by the contestant. 
 \end{remark}
 
 Given a phase number $s$ a color $c$  and an  edge  $e$ (belonging to $C_s$ or not),  the event of   {\em assignment of color $c$} for phase $s$ (CA-event, in short), notationally  $\mathrm{CA}(s, c; e)$, is the event that occurs  if the color assigned to $e$ by \valc at instant 
$t = \mathrm{time}(s; e) $  is $c$.  If $e$ belongs to the cycle $C_s$, i.e. if $e= e_s^i$ for some $i$,  then the CA-event $\mathrm{CA}(s, c; e)$ for $c= c_s^{((i+1) \mod 2) +1}$ is called the the   {\em correct color assignment event} (CCA-event, in short), notationally  $\mathrm{CCA}(C_s; e)$. In other words, a CCA-event occurs when the edge of its argument takes the color of the corresponding early edge of the   cycle of its argument. We refer to  the color $c_s^{((i+1) \mod 2) +1}$ for the edge $e_s^i$ as $e_s^i$'s ``correct" color  for phase $s$.

\remove{Given $C_s$ and an edge  $e$ (belonging to $C_s$ or not),  the   {\em correct color assignment event} (CCA-event, in short), notationally  $\mathrm{CCA}(C_s; e)$, is the event that occurs  if the color assigned by \valc at instant 
$t = \mathrm{time}(s; e) $ to $e$ is $c_s^{((i+1) \mod 2) +1}$.  In other words, a CCA-event occurs when the edge of its argument takes the color of the corresponding early edge of the   cycle of its argument. }

The chronological order of the edges $e_s^i$  of a cycle sequence $\mathcal{C}$ is the order induced by the ordering of $\mathrm{time}(s; e_s^i)$ (recall Remark \ref{rem:distincttime}).
\begin{definition}\label{def:anterior}
An {\em anterior}  conditional for an event $\mathrm{CA}(s,c; e)$ is any conjunction of events $\mathrm{CA}(s',c'; e')$, or negations of them,  such that $\mathrm{time}({s'}; e') <\mathrm{time}(s; e)$.  \end{definition}

 \begin{lemma}\label{lem:probccaevent}
The probability of any   CA-event, not referring to a cycle's early edge and its correct color, given any {\em anterior} conditional, is at most $$\frac{1}{\gamma(\Delta-1)+1}.$$
\end{lemma}
\begin{proof} 
It follows immediately from the fact that \valc assigns a color chosen uniformly at random from the set of colors that do not presently destroy 4-acyclicity and that this set, by Lemma 	\ref{sufficientcolors}, has cardinality at least 
$\gamma(\Delta-1)+1$, independently of past assignments (recall that  the total number of available colors is $K=\ceil{(2 + \gamma)(\Delta-1)} +1 $). 
\end{proof}

 \begin{remark}\label{rem:pivotal} [Identifying the pivotal edges with their corresponding early edges]\label{rem:pivotal} Observe that by definition, the CCA-events corresponding to the early edges of cycles occur with probability 1. To keep the notation simple, we will consider  CCA-events $\mathrm{CCA}(C_s; e_s^i)$ for all $i \geq 3$ but with the notational convention that  for $j=1,2$,   
 $\mathrm{CCA}(C_s; e_s^{l_s^j})$ stands for $\mathrm{CCA}(C_s; e_s^j)$, 
 i.e.   the  CCA-event of an early edge  is, by convention,  just  the CCA-event of the corresponding pivotal edge. 
 \end{remark}

\begin{lemma}  \label{lem:ourapproach} Given an admissible  sequence ${\mathcal S}= (e^1_s, e^2_s, k_s), s=1, \ldots, n$,  then for  the  probability of success  of \valc on input ${\mathcal  S}$ we have that:    
\begin{multline}\label{eq:bound} \Pr\left[ \text{\valc}  \text{\rm{ is successful on }} {\mathcal S}\right] \leq \\
 \left(\frac{1}{\gamma(\Delta-1)+1}\right)^n \prod_{s=1}^n\left(1 - \left(1 -\frac{1}{\gamma(\Delta-1)+1}\right)^{\Delta-1}\right)^{2k_s-3}. \end{multline}	
\end{lemma}

\begin{proof} 
Given a cycle-sequence ${\mathcal C}$ associated with $\mathcal{S}$, the edges incident on $v_s^i$ other 
than $e_s^{i-1}$ are called these edges stemming out $v_s^i$. Let these edges in chronological order be  $o^1, \ldots, o^{\delta_s^i}$,  where  $\delta_s^i$ denotes their cardinality (the $o^l$'s themselves depend on $s$ and $i$, however for brevity we omit the  corresponding indices).
 Obviously $\delta_s^i \leq \Delta-1$, 
where $\Delta$ is the max degree of the underlying graph $G$.  When $e_s^{i-1}$ is the previous to last edge of $C_s$ in its  positive traversal, we set $\delta_s^i =1$. Also, let the {\em parity} of a non-pivotal edge $e_s^i$ of  $C_s$ be 1 or 2 depending on whether $e_s^i$ is at odd distance from $e_s^1$ or $e_s^2$.

Let 
${\mathcal E}_{\mathcal C}^1, \ldots, {\mathcal E}_{\mathcal C}^{2\hat{k}}$ denote the 
CCA-events $\mathrm{CCA}(C_s; e_s^i), 1\leq s \leq  n$, $ i=3, \ldots, 2k_s$ 
of $\mathcal{C}$
in {\em chronological order} (by Remark~\ref{rem:distincttime}, we may assume that these events correspond to distinct time instants). Let  $\mathcal{A}_t$ denote the conditional that the chronologically first $t-1$ CCA-events of $\mathcal{C}$ hold ($\mathcal{A}_{t}$ depends on the execution). 

The {\em geometric order} of the edges of $\C$ is the one where the edges in each cycle are ordered as they are traversed in its positive traversal, whereas the cycles are ordered by their index (corresponding phase number).   Assume that the CCA-event for the $k$'th such edge ($k=1, \ldots, 2\hat{k}$)    in the geometric  order  is $\mathcal{E}_{\mathcal C}^{t_k}$. In other words,  $\mathcal{E}_{\mathcal C}^{t_k}, t=1, \ldots, 2\hat{k}$ is a reordering, in the sense of the geometric order,  of the chronological order $\mathcal{E}_{\mathcal C}^{t}, t=1, \ldots, 2\hat{k}$ (the reordering depends on the execution). 

We now define a random variable, denoted by  $\C$,  over the random color-choices of \valc.
$\C$ is the unique cycle-sequence such that for  every cycle $C_s$ of $\C$, every  non-pivotal edge $e_s^i$  is either the (unique) edge stemming out of $v_s^i$ that, at the beginning of phase $s$,  has the same color  as the pivotal edge of $C_s$ of the same parity, if there is such one,  or is the first edge stemming out of $v_s^i$ in some predetermined order of all edges of the underlying graph, otherwise. This is a well defined function over the space of random color-choices of \valc because we cannot have two homochromatic edges incident onto the same vertex. It is admittedly somewhat confusing that until now the notation $\C$ referred to an arbitrary but fixed  cycle-sequence, whereas in the sequel it denotes a  random variable depending on the random choices of \valc, yet we believe that we thus avoid overloading the notation. Once $\C$ is seen as an well-defined random variable, all the parameters defined up to now  in terms of a fixed  arbitrary $\C$, like $e_s^i$ (the $i$-th edge in the positive traversal of the $s$-th cycle of $\C$),   $\delta_s^i$ (the number of edges stemming out of the $v_s^i$), $o^l$ (the chronologically $l$-th edge stemming out of $v_s^i$) 
 etc. become well-defined random variables, depending only on the random choices of \valc and its input $\seq$. Also the events previously described  remain meaningful with $\C$ being the  random variable defined above instead of arbitrary and fixed. 
 
Now, observe that:
\begin{align}\label{eq:randomC}
\Pr\left[ \text{\valc}  \text{\rm{ is successful on }} {\mathcal S}\right]  & = 
\Pr\left[\C \text{\rm \ is bichromatic}\right]   \nonumber  \\ & =
  \prod_{t=1}^{2\hat{k}}\Pr\left[ {\mathcal E}_{\mathcal C}^t \mid {\mathcal A}_t\right]     =  \prod_{k=1}^{2\hat{k}} \Pr\left[ {\mathcal E}_{\mathcal C}^{t_k} \mid {\mathcal A}_{t_k}\right].	
\end{align}
Indeed, we obviously have that the event ``$\text{\valc}  \text{\rm{ is successful on }} {\mathcal S}$" is implied by the event ``$\C \text{\rm \ is bichromatic}$" ($\C$ is said to be  bichromatic if for all $s$, $C_s$ is bichromatic at the beginning of phase $s$). Also because no two homochromatic edges  can be incident to the same vertex, we have that the latter event is implied by the former, so we get the first equality of Equation \ref{eq:randomC}. The second equality follows by Remark \ref{rem:pivotal} and the third one is trivial.

  In order to find an upper bound for  $\Pr\left[ {\mathcal E}_{\mathcal C}^{t_k} \mid {\mathcal A}_{t_k}\right]$, first assume  that for some $s=1, \ldots, n, i=3, \ldots, 2k_s$, the event  ${\mathcal E}_{\mathcal C}^{t_k}$ is the event $\mathrm{CCA}(C_s;e_s^i)$. For brevity let $\delta_k$ denote $\delta_s^i$ and $c_k$ the correct color for $e_s^i$.

Let now $E^l$ be the event ``$(e_s^i = o^l) \land  \mathrm{CA}(s; c_k,  o^l)"$,  $l=1, \ldots, \delta_k$. Intuitively, $E^l$ means that $e_s^i$ is the chronologically $l$'th edge stemming out of $v_s^i$ and it gets the correct color. 
In case $e_s^i$ is a non-pivotal early edge of a cycle $C_s$ (i.e. the earliest edge among the same parity edges of $C_s$), then $E^l$ stands for ``$(e_s^i = o^l) \land  \mathrm{CA}(s; c_k,  e)"$, where $e$ is the pivotal edge $e_s^1$ or $e_s^2$ of the same parity as $e_s^i$ and $c_k$ is $e_s^i$'s color (recall Remark \ref{rem:pivotal}).
 
We have: 

\begin{equation}\label{eq:instantiations}
\Pr\left[ {\mathcal E}_{\mathcal C}^{t_k} \mid {\mathcal A}_{t_k} \right] =  
      \Pr\left[ \mathrm{CA}(s, c_k; e_s^i) \mid {\mathcal A}_{t_k} \right] =   \Pr\left[ \bigvee_{l=1}^{\delta_k} E^l \mid {\mathcal A}_{t_k}  \right], 
\end{equation}
because always (deterministically) $e_s^i$ is one of the edges that stem out of $v_s^i$.  Intuitively,  $\Pr\left[ \bigvee_{l=1}^{\delta_k} E^l \mid {\mathcal A}_{t_k}  \right]$ denotes the probability  that one of the edges stemming out of $e_s^i$ gets the correct color.

Observe that:
\begin{align}\label{eq:negation}
\Pr\left[ \bigvee_{l=1}^{\delta_k} E^l \mid {\mathcal A}_{t_k}  \right] &= 1 - 	\Pr\left[ \bigwedge_{l=1}^{\delta_k} (\neg E^l )\mid {\mathcal A}_{t_k}  \right] \nonumber \\&=  1-\prod_{l=1}^{{\delta}_k}\Pr\left[ \neg E^l \mid \left(\bigwedge_{m=1}^{l-1}\neg E^m\right) \land \mathcal{A}_{t_k}  \right].
\end{align}
Intuitively, the expression $\Pr\left[ \bigwedge_{l=1}^{\delta_k} (\neg E^l )\mid {\mathcal A}_{t_k}  \right]$ denotes the probability that none of the possible instantiations of $e_s^i$ (i.e. the edges that stem out of the   $v_s^i$)  gets the correct color.

Now observe that:
\begin{itemize}
\item[(i)]	\label{it:i}
for an $l$'s for which   $e_s^i \neq o^l$, it holds that $\Pr\left[ \neg E^l \mid \left(\bigwedge_{m=1}^{l}\neg E^m\right) \land \mathcal{A}_{t_k} 
 \right] = 1$, whereas 
\item[(ii)] \label{it:ii} for an $l$ for which $e_s^i = o^l$, it holds that 
 $$ \Pr\left[ \neg E^l \mid \left(\bigwedge_{m=1}^{l-1}\neg E^m\right) \land \mathcal{A}_{t_k}  \right] \geq 1 - \frac{1}{\gamma(\Delta-1)+1},$$ because when $e_s^i = o^l$, the event  $E^l$ is equivalent to $\mathrm{CA}(s, c_k; o^l)$ and the conditional 
 $$\left(\bigwedge_{m=1}^{l-1}\neg E^m\right) \land \mathcal{A}_{t_k}$$ is equivalent to  
 $\mathcal{A}_{t_k}$, which   is anterior to $\mathrm{CA}(s, c_k; o^l)$, because $e_s^i = o^l$. 
\end{itemize} 
Therefore for all $l$, we have that  $$ \Pr\left[ \neg E^l \mid \left(\bigwedge_{m=1}^{l-1}\neg E^m\right) \land \mathcal{A}_{t_k}\right]  \geq 1 - \frac{1}{\gamma(\Delta-1)+1},$$

 Therefore: 
 $$\Pr\left[ \bigwedge_{l=1}^{\delta_k} (\neg E^l) \mid {\mathcal A}_{t_k}  \right] \geq \left(1 - \frac{1}{\gamma(\Delta-1)+1}\right)^{d_k},$$
  where $d_k =1 $  if $k$ corresponds to  the last (in the geometric order) edge $e_s^i$ of $C_s$, and $d_k = \Delta-1$ otherwise ($d_k$ is {\em not} random).

Therefore, from \eqref{eq:instantiations} and  \eqref{eq:negation} we get that:
 $$  \Pr\left[ {\mathcal E}_{\mathcal C}^{t_k} \mid {\mathcal A}_{t_k} \right]  \leq 1 -  \left(1 - \frac{1}{\gamma(\Delta-1)+1}\right)^{d_k }.$$
Observe that among the   $k= 1, \ldots, \hk$ 
there are $n$ values of $k$ for which  $  d_k =1 $ (the $k$'s for which $e_s^i$ is the closing edge of $C_s$), whereas for the remaining  $\sum_{s=1}^n (2k_s -3)$ values of $k$,  we have $ d_k = \Delta-1$. 
So we can conclude from \eqref{eq:randomC} that    that $\Pr\left[ \text{\valc}  \text{\rm{ is successful on }} {\mathcal S}\right]$ 
  is bounded from above by: \begin{equation*}\left(\frac{1}{\gamma(\Delta-1)+1}\right)^n \prod_{s=1}^n\left(1 - \left(1 -\frac{1}{\gamma(\Delta-1)+1}\right)^{\Delta-1}\right)^{2k_s-3}. \qedhere \end{equation*}\end{proof}
\begin{remark}\label{rem:vertexc}
The fact that no two homochromatic edges can be incident onto the same vertex is essential for the correctness of the above proof. For example in the case of vertex coloring, where we can have more than one homochromatic vertices neighboring with the same vertex, the proof does not go through. Indeed, even if we define the random variable $\C$ to be  the unique cycle-sequence such that for  every cycle $C_s$ of $\C$, every  non-pivotal edge $e_s^i$  is either the first  edge stemming out $v_s^i$ that has the same color as the pivotal edge of $C_s$ of the same parity, if there is one,  or is the first edge stemming out of $v_s^i$, if there is none, still Equation \ref{eq:randomC}	can only take the form:
\begin{align}\label{eq:randomC2}
\Pr\left[ \text{\valc}  \text{\rm{ is successful on }} {\mathcal S}\right]  & \geq
\Pr\left[\C \text{\rm \ becomes bichromatic}\right]   \nonumber  \\ & = 
  \prod_{t=1}^{2\hat{k}}\Pr\left[ {\mathcal E}_{\mathcal C}^t \mid {\mathcal A}_t\right]     =  \prod_{k=1}^{2\hat{k}} \Pr\left[ {\mathcal E}_{\mathcal C}^{t_k} \mid {\mathcal A}_{t_k}\right].	
\end{align}\end{remark}
\begin{corollary}\label{cor:ourapproach}
If ${\mathcal S}$ is as in Lemma \ref{lem:ourapproach}  then  
\begin{multline}\label{eq:ourapproach} 
\Pr\left[ \text{\valc}  \text{\rm{ is successful on }} {\mathcal S}\right]
\leq  \\ \left(\frac{1}{\Delta-1}\right)^n  \prod_{i=1}^n \left(\frac{1}{\gamma}\left(1 -  e^{-\frac{1}{\gamma}} \right)^{2k_i-3}\right).
\end{multline}
\end{corollary}
\begin{proof} 
Use the  inequality $$1-x>e^{-\frac{x}{1-x}}, \forall x<1, x \neq 0$$  (see e.g. \cite[inequality (4.5.7)]{NIST:DLMF})      
for $x=\frac{1}{\gamma(\Delta-1)+1}$, which implies, after elementary operations,  that $\left(1 -\frac{1}{\gamma(\Delta-1)+1}\right)^{\Delta-1}$  is at least $e^{-\frac{1}{\gamma}}$.
\end{proof}

We now turn to the estimation of the probability, call it $\hat{P}_n$,   that  \cola lasts for {\em at least}  $n$ phases.  
\begin{theorem} \label{thm:recurence} The expression $P_n$ defined as

\begin{equation}\label{eq:colfinalbound}
 {P}_n = \sum_{\substack{n_1+\cdots+n_{m} = n \\ n_1, \ldots, n_{m} \geq 0}} Q_{n_1} \cdots Q_{n_m},
\end{equation}
where $Q_n$ is defined by the recursion:
\begin{equation}
\label{eq:colsumvalid}
Q_n =  \sum_{k \geq 3}   \Bigg( \left(\frac{1}{\gamma}\left(1 -  e^{-\frac{1}{\gamma}} \right)^{2k-3}\right)  \cdot \sum_{\substack{n_1+\cdots+n_{2k} = n-1 \\ n_1, \ldots, n_{2k} \geq 0}} Q_{n_1} \cdots Q_{n_{2k}}\Bigg), 
\end{equation}
with $Q_0=1$, is an upper bound of $\hat{P}_n$.	
\end{theorem}
\begin{proof}
Let $f$ be an {\em unlabeled}  ordered forest with $m$ trees (recall, $m$ is the number of edges of $G$) and with $n$ internal nodes,  each, considered in their pre-order,  having an even degree $2k_s, s=1, \ldots, n$ ($k_s \geq 3$). Let $p_f$ be the probability that the witness forest
   of an execution of \cola on $G$ has underlying unlabeled forest $f$. 
Obviously \begin{equation}\label{eq:hatbound}\hat{P}_n =  \sum_{f} p_f,\end{equation} where  the sum  above ranges over unlabeled forests as described. 

With each internal node of $f$ with half-degree $k_s$, we associate  weight $$\frac{1}{\gamma}\left(1 -  e^{-\frac{1}{\gamma}} \right)^{2k_s-3}, $$ and we associate weight 1 with the leaves of $f$. We define: \begin{equation}\label{eq:norm} \lVert f \rVert = \prod_{i=1}^n \left(\frac{1}{\gamma}\left(1 -  e^{-\frac{1}{\gamma}} \right)^{2k_i-3}\right).\end{equation} We now prove:
\paragraph{Claim} $ p_f \leq \lVert f \rVert$.

\bigskip\noindent{\em Proof of Claim.} Let ${\mathcal L} = (e_1^1, C_1), \ldots, (e_n^1, C_n)$ be a sequence which if considered as label-sequence of $f$ leads to  an  $n$-witness forest for \cola. We will  examine what the possibilities for ${\mathcal L}$ are and thus bound the probability that at least one such label-sequence, together with $f$,  comprises a witness of \cola. 

First observe that the lengths of the cycles $C_s$ should coincide with the positive degrees $2k_s$ of $f$, therefore the lengths of the $C_s$ of ${\mathcal L}$ are uniquely determined. Then observe that the edge-label $e_1^1$ of ${\mathcal L}$,  which is to be assigned to the first $r$  root of $f$ that is not a leaf,  is uniquely determined to be the $k$'th edge of $G$, where $k-1$ is the number of leaf-roots that precede $r$ (recall the definition of a witness forest, the way it is defined by \cola and the fact that we introduced isolated roots for feasible forests in order  to cover all edges of $G$). 
There are  at most $(\Delta -1)$ choices for $e_1^2$, the second pivotal edge of $C_1$. 
This introduces a factor of $(\Delta -1)$ for the sought after upper bound of $p_f$. Once $e_1^1, e_1^2$ of ${\mathcal L}$ are determined,  the event that $f$, together with such an  ${\mathcal L}$,  gives a witness for \cola  necessitates that  some $C_1$ with pivotal edges $e_1^1, e_1^2$ and length $2k_1$ is bichromatic.  
For each such possibility for $C_1$, the second edge-label $e_2^1$ of ${\mathcal L}$ is determined to be the first edge of $C_1$ that does not correspond to a leaf of $f$ (recall again  the definition of a witness forest, the way it is defined by \cola and the fact that we introduced isolated leaves  for feasible forests in order to cover all edges of each cycle-label of each internal node). Again there are at most $(\Delta-1)$ possibilities for $e_2^2$, thus introducing a factor of $(\Delta-1)$ to the sought after upper bound. 
And again, it becomes necessary that  some $C_2$ with with pivotal edges $e_2^1, e_2^2$ and length $2k_2$ is bichromatic. We continue in this fashion, following the structure of $f$: when a leaf is reached,  we go back to the last internal node (in a depth-first fashion), then the  next edge of this internal node's cycle-label should be the next edge-label in ${\mathcal L}$; when a tree of the forest $f$  is exhausted, we go to the next root,  then the next edge of $G$  should be the next edge-label ${\mathcal L}$. Now from Proposition \ref{prop:proposition}, which bounds the probability that  at least one forest  with labels $(e_s^1, C_s)_s$ is a witness for \cola,  and Corollary \ref{cor:ourapproach}, which bounds the probability that \valc is successful on input $(e_s^1, e_s^2, k_s)_s$, and taking into account the $n$ factors $(\Delta -1)$, which correspond to the possible choices of $e_s^2$, we conclude  that $p_f$ is bounded from above by the product of the weights: 
$$\lVert f \rVert = \prod_{i=1}^n \left(\frac{1}{\gamma}\left(1 -  e^{-\frac{1}{\gamma}} \right)^{2k_i-3}\right).$$
This concludes the proof of the Claim and we now return to the proof of Theorem~\ref{thm:recurence}.

Consider  {\em unlabeled} ordered forests with $m$ ordered trees  with  $n$  internal  nodes in total, whose out-degrees are even and $\geq 6$; assign  weight $w_{2k}$ to all internal nodes of out-degree $2k$; also assign weight 1 to all leaves and let the weight of such a tree be the product of the weight assigned to all its nodes. Then by distributivity of multiplication over addition it follows that the sum $f_n$ of the weights of all such weighted unlabeled ordered forests is given by: 

\begin{equation}\label{eq:colfinalboundlow}
 {f}_n = \sum_{\substack{n_1+\cdots+n_{m} = n \\ n_1, \ldots, n_{m} \geq 0}} t_{n_1} \cdots t_{n_m},
\end{equation}
where $t_n$, the sum of the weights  of  weighted unlabeled ordered rooted {\em trees} with $n$ internal nodes,  is given  by the recurrence:
\begin{equation}
\label{eq:colsumvalidlow}
t_n =  \sum_{k \geq 3}  w_{2k} \Bigg(   \sum_{\substack{n_1+\cdots+n_{2k} = n-1 \\ n_1, \ldots, n_{2k} \geq 0}} t_{n_1} \cdots t_{n_{2k}}\Bigg), 
\end{equation}
with $t_0=1$. By the Claim above we get that: 
\begin{equation}\label{eq:pf}
 \sum_{f} p_f  \leq  \sum_{f}\lVert f \rVert = \sum_{f}\left[ \prod_{i=1}^n \left(\frac{1}{\gamma}\left(1 -  e^{-\frac{1}{\gamma}} \right)^{2k_i-3}\right)\right].  
\end{equation}  
The required  follows from \eqref{eq:hatbound} and \eqref{eq:pf}  making use of the recurrence given by \eqref{eq:colfinalboundlow} and \eqref{eq:colsumvalidlow}.
\end{proof}
\subsection{Asymptotic analysis of the recurrence}
We will asymptotically analyze the coefficients of the OGF $Q(z)$  of $Q_n$. Towards this end, multiply both sides  of equality  \eqref{eq:colsumvalid} 
by $z^n$ and sum for $n= 1, \dots, \infty$ to get 
\begin{equation}
\label{Q}
Q(z) -1 = \sum_{k\geq 3}\left[\frac{1}{\gamma}\left(1-e^{-\frac{1}{\gamma}}\right)^{2k-3}zQ(z)^{2k}  \right], 
\end{equation}
with $Q(0) =1$. Setting $W(z) = Q(z)-1$ we get 
\begin{equation}
\label{W}
W(z) = \sum_{k\geq 3}\left[ \frac{1}{\gamma}\left(1-e^{-\frac{1}{\gamma}}\right)^{2k-3}z(W(z)+1)^{2k}  \right], 
\end{equation}
with $W(0) =0$. For notational convenience, set  $W = W(z)$. Then from \eqref{W} we get:
\begin{equation}
\label{W2}
W = z\frac{1}{\gamma}\cdot \frac{\left(1-e^{-\frac{1}{\gamma}}\right)^3(W+1)^6} {1-\left(1-e^{-\frac{1}{\gamma}}\right)^2(W+1)^2}.
\end{equation}
Set now 
\begin{equation}
\label{defphi}
\phi(x) =\frac{1}{\gamma}\cdot \frac{\left(1-e^{-\frac{1}{\gamma}}\right)^3(x+1)^6}{1-\left(1-e^{-\frac{1}{\gamma}}\right)^2(x+1)^2},
\end{equation}
to get from \eqref{W2}:
\begin{equation}
\label{phi}
W = z\phi(W).
\end{equation}
Let now $R= \frac{1}{\left(1-e^{-\frac{1}{\gamma}} \right)}-1$  be the radius of convergence of the series representing  $\phi$ at $0$.   Let also $\tau$ be the (necessarily unique) solution  in the interval $(0,R)$ of the characteristic equation  (in $\tau$):
\begin{equation}
\label{FSequation}
\phi(\tau) - \tau\phi'(\tau) =0.
\end{equation}
Finally, let \begin{equation}\label{rho}\rho =  \frac{\tau}{\phi(\tau)}.\end{equation}
By \cite[Proposition IV.5]{Flajolet:2009:AC:1506267} (it is trivial to check that the hypotheses in that Theorem are satisfied for $\gamma >0$),
we get $[z^n]Q \bowtie {(1/\rho)}^n$, i.e. $\limsup \ ([z^n]Q)^{1/n} = 1/\rho$ 
 (see \cite[IV.3.2]{Flajolet:2009:AC:1506267}).

Now by a simple search (through Maple, for the code see \cite{GKPT2}) we found that for $\gamma=1.73095$, the unique positive solution of \eqref{FSequation} in the radius of convergence is $\tau=0.1747094762$, and this value of $\tau$ gives  $1/\rho=0.9999789027<1$. Therefore by making use of  \eqref{eq:colfinalbound}, we get:

\begin{theorem}
\label{colthm}
Assuming  $\Delta$, the maximum degree of the graph $G$, is  constant, and given the availability of at least   $3.74(\Delta-1)+1$  colors,  there exists an integer $N$, which depends  linearly on $m$, the number of edges of $G$, and a constant $\rho >1 $ such that if $n/\log n \geq N$
then  the probability that \emph{\cola} executes at least  $n$ calls of \emph{\reco} is $< (1/\rho)^n$; therefore the graph has  an acyclic edge coloring. 
\end{theorem}
Now if the graph has girth $2r-1$ for $r\geq 4$, the previous arguments carry over with minimal changes. Namely,
equation \eqref{eq:colsumvalid} becomes: 
\begin{equation}
\label{gcolsumvalid}
Q_n =  \sum_{k \geq r} \Bigg[\frac{1}{\gamma}\left(1-e^{-\frac{1}{\gamma}}\right)^{2k-3}\cdot  \sum_{\substack{n_1+\cdots+n_{2k} = n-1 \\ n_1, \ldots, n_{2k} \geq 0}} Q_{n_1} \cdots Q_{n_{2k}}\Bigg], \quad Q_0 =1. 
\end{equation} 
Also in \eqref{Q} and \eqref{W}, the starting point of the summation is changed from $3$ to $r$. Moreover,   equation \eqref{W2}  becomes:
\begin{equation}
\label{gW2}
W = z\frac{1}{\gamma}\cdot \frac{\left(1-e^{-\frac{1}{\gamma}}\right)^{2r-3}(W+1)^{2r}}{1-\left(1-e^{-\frac{1}{\gamma}}\right)^2(W+1)^2},
\end{equation}
and equation \eqref{defphi}  becomes:
\begin{equation}
\label{gdefphi}
\phi(x) = \frac{1}{\gamma}\cdot \frac{\left(1-e^{-\frac{1}{\gamma}}\right)^{2r-3}(x+1)^{2r}}{1-\left(1-e^{-\frac{1}{\gamma}}\right)^2(x+1)^2}. 
\end{equation}
Working as before, we get numerical results depicted in Figure~\ref{fig:girth} with sample specific values explicitly given in Figure~\ref{fig:results}.
\begin{figure}[t]
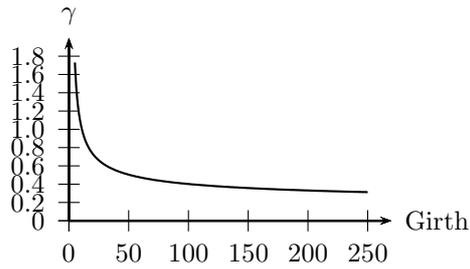

\centering
\savedata{\mydata}[
  {{5, 1.731000000}, {6, 1.488000000}, {7, 1.326000000}, {8, 1.210000000}, {9, 1.121000000}, {10, 1.051000000}, {11, .9930000000}, {12, .9460000000}, {13, .9050000000}, {14, .8700000000}, {15, .8400000000}, {16, .8130000000}, {17, .7890000000}, {18, .7670000000}, {19, .7480000000}, {20, .7300000000}, {21, .7140000000}, {22, .6990000000}, {23, .6860000000}, {24, .6730000000}, {25, .6610000000}, {26, .6500000000}, {27, .6400000000}, {28, .6300000000}, {29, .6210000000}, {30, .6120000000}, {31, .6040000000}, {32, .5970000000}, {33, .5890000000}, {34, .5820000000}, {35, .5760000000}, {36, .5700000000}, {37, .5640000000}, {38, .5580000000}, {39, .5520000000}, {40, .5470000000}, {41, .5420000000}, {42, .5370000000}, {43, .5320000000}, {44, .5280000000}, {45, .5240000000}, {46, .5190000000}, {47, .5150000000}, {48, .5120000000}, {49, .5080000000}, {50, .5040000000}, {51, .5010000000}, {52, .4970000000}, {53, .4940000000}, {54, .4910000000}, {55, .4880000000}, {56, .4840000000}, {57, .4820000000}, {58, .4790000000}, {59, .4760000000}, {60, .4730000000}, {61, .4710000000}, {62, .4680000000}, {63, .4650000000}, {64, .4630000000}, {65, .4610000000}, {66, .4580000000}, {67, .4560000000}, {68, .4540000000}, {69, .4520000000}, {70, .4490000000}, {71, .4470000000}, {72, .4450000000}, {73, .4430000000}, {74, .4410000000}, {75, .4400000000}, {76, .4380000000}, {77, .4360000000}, {78, .4340000000}, {79, .4320000000}, {80, .4310000000}, {81, .4290000000}, {82, .4270000000}, {83, .4260000000}, {84, .4240000000}, {85, .4220000000}, {86, .4210000000}, {87, .4190000000}, {88, .4180000000}, {89, .4160000000}, {90, .4150000000}, {91, .4140000000}, {92, .4120000000}, {93, .4110000000}, {94, .4090000000}, {95, .4080000000}, {96, .4070000000}, {97, .4060000000}, {98, .4040000000}, {99, .4030000000}, {100, .4020000000}, {101, .4010000000}, {102, .3990000000}, {103, .3980000000}, {104, .3970000000}, {105, .3960000000}, {106, .3950000000}, {107, .3940000000}, {108, .3930000000}, {109, .3920000000}, {110, .3910000000}, {111, .3900000000}, {112, .3890000000}, {113, .3880000000}, {114, .3870000000}, {115, .3860000000}, {116, .3850000000}, {117, .3840000000}, {118, .3830000000}, {119, .3820000000}, {120, .3810000000}, {121, .3800000000}, {122, .3790000000}, {123, .3780000000}, {124, .3770000000}, {125, .3760000000}, {126, .3760000000}, {127, .3750000000}, {128, .3740000000}, {129, .3730000000}, {130, .3720000000}, {131, .3710000000}, {132, .3710000000}, {133, .3700000000}, {134, .3690000000}, {135, .3680000000}, {136, .3670000000}, {137, .3670000000}, {138, .3660000000}, {139, .3650000000}, {140, .3650000000}, {141, .3640000000}, {142, .3630000000}, {143, .3620000000}, {144, .3620000000}, {145, .3610000000}, {146, .3600000000}, {147, .3600000000}, {148, .3590000000}, {149, .3580000000}, {150, .3580000000}, {151, .3570000000}, {152, .3560000000}, {153, .3560000000}, {154, .3550000000}, {155, .3540000000}, {156, .3540000000}, {157, .3530000000}, {158, .3530000000}, {159, .3520000000}, {160, .3510000000}, {161, .3510000000}, {162, .3500000000}, {163, .3500000000}, {164, .3490000000}, {165, .3480000000}, {166, .3480000000}, {167, .3470000000}, {168, .3470000000}, {169, .3460000000}, {170, .3460000000}, {171, .3450000000}, {172, .3450000000}, {173, .3440000000}, {174, .3430000000}, {175, .3430000000}, {176, .3420000000}, {177, .3420000000}, {178, .3410000000}, {179, .3410000000}, {180, .3400000000}, {181, .3400000000}, {182, .3390000000}, {183, .3390000000}, {184, .3380000000}, {185, .3380000000}, {186, .3370000000}, {187, .3370000000}, {188, .3370000000}, {189, .3360000000}, {190, .3360000000}, {191, .3350000000}, {192, .3350000000}, {193, .3340000000}, {194, .3340000000}, {195, .3330000000}, {196, .3330000000}, {197, .3320000000}, {198, .3320000000}, {199, .3320000000}, {200, .3310000000}, {201, .3310000000}, {202, .3300000000}, {203, .3300000000}, {204, .3290000000}, {205, .3290000000}, {206, .3290000000}, {207, .3280000000}, {208, .3280000000}, {209, .3270000000}, {210, .3270000000}, {211, .3270000000}, {212, .3260000000}, {213, .3260000000}, {214, .3250000000}, {215, .3250000000}, {216, .3250000000}, {217, .3240000000}, {218, .3240000000}, {219, .3230000000}, {220, .3230000000}, {221, .3230000000}, {222, .3220000000}, {223, .3220000000}, {224, .3220000000}, {225, .3210000000}, {226, .3210000000}, {227, .3210000000}, {228, .3200000000}, {229, .3200000000}, {230, .3190000000}, {231, .3190000000}, {232, .3190000000}, {233, .3180000000}, {234, .3180000000}, {235, .3180000000}, {236, .3170000000}, {237, .3170000000}, {238, .3170000000}, {239, .3160000000}, {240, .3160000000}, {241, .3160000000}, {242, .3150000000}, {243, .3150000000}, {244, .3150000000}, {245, .3140000000}, {246, .3140000000}, {247, .3140000000}, {248, .3140000000}, {249, .3130000000}, {250, .3130000000}}]
\psset{xAxisLabel={Girth},yAxisLabel={$\gamma$}}
\begin{psgraph}[Dx=50,xlabelOffset=0.0,Dy=0.2]{->}(0,0)(270,2){0.35\textwidth}{0.2\textwidth}
\listplot[plotstyle=curve,dotstyle=o,fillcolor=red]{\mydata}
\end{psgraph}
\bigskip
\caption{$\gamma$ as a function of girth}
\label{fig:girth}
\end{figure}

\section{Discussion} \label{sec: discussion} There are several conceivable possibilities for improvement of the $\lceil 3.74(\Delta-1)\rceil+1 $ bound. For example, in \cola we can recolor not all edges of $C_s$ but all but two consecutive of them (because it is only those that determine bichromaticity). Also in Theorem  \ref{thm:recurence}, we can conceivably improve the bound provided by claiming that  for each execution, there is only one $k_s$ so that the cycle $C_s$ of half-length $k_s$ is bichromatic. Nevertheless we opted not to consider these possible improvements, as our aim was only to present a  probabilistic analysis of a Moser-type algorithm with dependent variables. 

\section*{Acknowledgements}
We are grateful to Juanjo Ru\'{e} for showing to us how to deal with the asymptotics of the coefficients of inverse generating functions. We are grateful to  Dieter Mitsche on one hand and to  Dimitris Achlioptas and  Fotis Iliopoulos, on the other, for pointing out   errors in previous versions of this paper;  the second and fourth authors are also  indebted to the latter two   for   initiating them to this line of research. The contributions of the undergraduate students   Rafail Kartsioukas, Georgios Kontogeorgiou, Stavros Messaris and  Zoi Terzopoulou at  various stages of this work were crucial. We sincerely thank them all. We are also much indebted  to Gwena\"{e}l Joret for pointing out the incompleteness of the proof of \cite[Lemma 5]{giotis2017acyclic} and for the illuminating for us exchange of messages that followed. John Livieratos' remarks at the last phases of this work were very helpful.

\section*{References}


\end{document}